\documentclass{eptcs}
\providecommand{\event}{GANDALF 2017} 
\usepackage{breakurl}             
\usepackage{underscore}           

\usepackage{amssymb}
\usepackage{graphicx}
\usepackage{xspace}
\usepackage{paralist} 
\usepackage[normalem]{ulem} 
\usepackage{amsmath}
\usepackage{todonotes}

\newtheorem{theorem}{Theorem}

\newenvironment{proof}{\textit{Proof.}}{\hfill$\diamond$}

\newcommand{\ignore}[1]{}

\newcommand{\realsp}{\ensuremath{\mathbb{R}^+_0}}

\newcommand{\naturals}{\ensuremath{\mathbb{N}}}

\newcommand{\mktuple}[1]{\ensuremath{\langle #1 \rangle}}

\newcommand{\zero}{\textbf{0}}

\newcommand{\atfuture}{at next}
\newcommand{\atpast}{at last}
\newcommand{\atF}{\ensuremath{@F}}
\newcommand{\atP}{\ensuremath{@P}}
\newcommand{\atsF}{\ensuremath{@\tilde{F}}}
\newcommand{\atsP}{\ensuremath{@\tilde{P}}}

\newcommand{\default}{\ensuremath{\mathit{def}}}
\newcommand{\ite}{\ensuremath{\mathit{ite}}}

\title{Linear-time Temporal Logic\\ with Event Freezing
  Functions%
\thanks{This work has received funding from the European Union's Horizon 2020
research and innovation programme under the Grant
Agreement No. 700665 (project CITADEL).}
}
\author{Stefano Tonetta
\institute{FBK-irst}
\email{tonettas@fbk.eu}}
\def\titlerunning{Linear-time Temporal Logic with Event Freezing Functions}
\def\authorrunning{S. Tonetta}

\begin{document}

\maketitle
\begin{abstract}
Formal properties represent a cornerstone of the system-correctness
proofs based on formal verification techniques such as model checking.
Formalizing requirements into temporal properties may be very complex
and error prone, due not only to the ambiguity of the textual
requirements but also to the complexity of the formal
language. Finding a property specification language that balances
simplicity, expressiveness, and tool support remains an open
problem in many real-world contexts.

In this paper, we propose a new temporal logic, which extends
First-Order Linear-time Temporal Logic with Past adding two operators
``\atfuture'' and ``\atpast'', which take in input a term and a
formula and represent the value of the term at the next state in the
future or last state in the past in which the formula holds. We
consider different models of time (including discrete, dense, and
super-dense time) and Satisfiability Modulo Theories (SMT) of the
first-order formulas. The ``\atfuture'' and ``\atpast'' functions can
be seen as a generalization of Event-Clock operators and can encode
some Metric Temporal operators also with counting. They are useful to
formalize properties of component-based models because they allow to
express constraints on the data exchanged with messages at different
instants of time. We provide a simple encoding into equisatisfiable
formulas without the extra functional symbols. We implement a
prototype tool support based on SMT-based model checking.
\end{abstract}

\section{Introduction}

The specification of properties is a fundamental step in the formal
verification process. System requirements must be captured by formal
properties, typically using logic formulas. However, this is often a
complex activity and may become a blocking factor for an industrial
adoption of the formal techniques. The informal requirements are quite
ambiguous but also the complexity of the target logic may be the
source of errors in the specification. Finding a property
specification language that balances simplicity, expressiveness, and
analysis tool support remains an open problem in many real-world
contexts.

One of the most popular logics used in computer science to specify
properties for formal verification is Linear-time Temporal Logic (LTL)
\cite{Pnu77}. The model of time is typically discrete and models are
discrete, linear sequences of states. We consider First-Order LTL
\cite{MP92} with future as well as past
operators~\cite{LichtensteinPZ85}. Thus, the system state is described
by individual variables and first-order functions/predicates can
describe their relationship. In the spirit of Satisfiability Modulo
Theories (SMT)~\cite{BarrettSST09}, the formulas are interpreted
modulo a background first-order theory as in~\cite{GhilardiNRZ07}
(here, we restrict to a quantifier-free fragment with all signature's
symbols rigid). Efficient SMT-based model checking techniques can be
used to verify temporal properties on systems described with first-order
formulas (see, e.g, \cite{CimattiGMT14,DanielCGTM16}).

In the case of real-time systems, LTL is interpreted over a dense
model of time or super-dense (dense time with possible sequences of
instantaneous events, as needed for example for asynchronous real-time
systems).  When considering real models of time, it becomes natural to
have constraints on the time elapsing between different events.
Therefore, LTL has been extended either with clocks/freezing operators
as in TPTL~\cite{AlurH93} or with metric operators as in
\cite{Koymans90,AFH96,RaskinS99}. Again, these extensions can be
combined with first-order logic (e.g., to represent message
passing~\cite{Koymans92} or for monitoring specification
in~\cite{BasinKlaedtkeMuller}).

When a system or component is seen as a black box, the properties must
be specified in terms of the observable variables or messages
exchanged with the system environment. This is for example the case of
properties of monitors (which trigger alarms based on some condition
on the observed variables/messages) or contract-based specifications
(which formalize the assumptions/guarantees of components independently
of the implementation). In these cases, the properties must capture
the relationship between the observable variables at different points
of time, without referring to internal variables that store the
corresponding values. It is therefore necessary to have suitable
mechanisms to refer to the value of variables at different points of
time. Instead of enriching the specification language with registers
as in register automata~\cite{DemriL09} to explicitly store the value
of variables in an operational-style specification, we adopt a more
declarative style with functions that directly return the value of
variables at the next or last state in which a formula will be/was
true.

More specifically, we extend the quantifier-free fragment of First-Order
Linear-Time Temporal Logic with Past operators adding ``\atfuture''
$u\atsF(\phi)$ and ``\atpast'' $u\atsP(\phi)$ functional symbols, which
are used to represent the value of a term $u$ at the next state in the
future or at the last state in the past in which a formula $\phi$ holds. For
example, the formula $G(alarm\leftrightarrow
x\atsP(read)=x\atsP^2(read))$ says that $alarm$ is true iff in the
last two points in which $read$ was true the variable $x$ had the same
value.
We consider different models of time, including discrete,
dense, and super-dense time. In the dense time setting, the definition
has to take into account that a minimum time point may not exist
because a formula may be true on open intervals.
The ``\atfuture'' and ``\atpast'' functions can be seen as a
generalization of Event-Clock Temporal Logic (ECTL)
operators~\cite{RaskinS97,HenzingerRS98, RaskinS99} (which, on turn,
are the logical counterpart of event clocks~\cite{AlurFH99}) and can
encode some Metric Temporal Logic (MTL) operators~\cite{Koymans90}
also with counting~\cite{HirshfeldR06,OrtizLS10}. They are useful to
formalize properties of component-based models because
they allow us to express constraints on the data exchanged with
messages at different instants of time. We provide a simple encoding
of the formulas with these extra functional symbols into
equisatisfiable formulas without them. We implemented a prototype tool
support based on SMT-based model checking.

The natural alternative to the logic we proposed would be to use
registers and freezing quantifiers as in \cite{DemriL09} and
TPTL. Despite freezing quantifiers provide a higher expressiveness
(also with respect to MTL~\cite{BCM10}), they are not so common
in industrial applications (at least compared to LTL and MTL), either
because they are less intuitive to use or they lack tool support.

The main contributions of the paper are the following. First, we
identify an extension of LTL that can express interesting properties
relating variables at different points of time. Second, we define the
new operators in a very rich setting that includes first-order
constraints, past operators, dense and super-dense semantics; this
gives also a uniform treatment of LTL satisfiability modulo theories
in the case of real time models. Third, we provide a prototype tool
support that effectively proves interesting properties, while many
logics in the real-time setting lack of tool support.

The rest of the paper is organized as follows: Section~\ref{sec:bg}
introduces the considered time models, LTL satisfiability modulo
theories, and its extension with metric operators;
Section~\ref{sec:main} defines the extension with the new event
freezing functions; Section~\ref{sec:sat} describes the satisfiability
procedure; Section~\ref{sec:eval} presents some preliminary
experimental results; finally, Section~\ref{sec:conc} concludes the
paper and draws directions for future work.

\section{Background}
\label{sec:bg}

\subsection{Time models}

\realsp is the set of non-negative real numbers. A time interval is a
convex subset of \realsp. The left endpoint of an interval $I$ is
denoted by $l(I)$, while the right endpoint by $r(I)$. Two intervals
$I$ and $I'$ are almost adjacent iff $r(I)=l(I)$ (so they may overlap
in at most one point). A singular interval is an interval in the form
$[a,a]$ for some $a\in\realsp$. A time interval sequence is a sequence
$I_0,I_1,I_2,\ldots$ of time intervals such that, for all $i\geq 0$,
$I_i$ and $I_{i+1}$ are almost adjacent and $\bigcup_{i\geq
  0}I_i=\realsp$.

We consider different models of time~\cite{AH91,AlurH93}. A time model is a
structure $\tau=\mktuple{T,<,\zero,v}$ with a domain $T$, a total
order $<$ over $T$, a minimum element $\zero\in T$, and a function
$v:T\rightarrow \realsp$ that represents the real time of a time point
in $T$. The $v$ function is used instead of a distance (e.g., as in
\cite{Koymans90}) to treat the weakly-monotonic case in a more uniform
way.
A \emph{time point} is an element of $T$. In particular, we
consider the following models:
\begin{itemize}
\item
discrete time models where $T=\naturals$, $\zero$ and $<$ are the
standard zero and order over natural numbers, $v(0)=0$ and
$v(0),v(1),v(2),\ldots$ is a non-decreasing divergent sequence (this
is also called the pointwise semantics; we use these models also for
discrete-time LTL ignoring these real-time timestamps);
\item
dense (strictly-monotonic) time model where $T=\realsp$, $\zero$ and
  $<$ are the standard zero and order over the real numbers,
  and $v$ is the identity function;
\item
super-dense (weakly-monotonic) time models where 1)
$T\subset\naturals\times\realsp$ such that the sequence of sets
$I_0,I_1,I_2,\ldots$ where, for all $i\geq 0$, the set $I_i:=\{t\mid
\mktuple{i,t}\in T\}$, is a time interval sequence (thus subsequent
intervals can overlap in at most one point), 2)
$\mktuple{i,t}<\mktuple{i',t'}$ iff $i<i'$ or $i=i'$ and $t<t'$, 3)
$\zero=\mktuple{0,0}\in\naturals\times\realsp$, and 4)
$v(\mktuple{i,t})=t$.
\end{itemize}

\subsection{First-Order Linear-time Temporal Logic}

We consider First-Order Linear-time Temporal Logic with Past
Operators, which we refer to for simplicity as LTL.

Given a first-order signature $\Sigma$ and a set $V$ of variables, we
define the syntax of $\Sigma$-formulas as follows:
\begin{align*}
&\phi:=p(u,\ldots,u)\mid \phi\wedge\phi\mid \neg\phi\mid \phi
  \tilde{U}\phi\mid \phi \tilde{S}\phi\\
&u:=c\mid x\mid f(u,\ldots,u)
\end{align*}
\noindent
where $p$ is a predicate symbol of $\Sigma$, $u$ is a term, $f$ is
a functional symbol of $\Sigma$, $c$ is a constant symbol of $\Sigma$,
and $x$ is a variable in $V$.

$\Sigma$-formulas are interpreted by a first-order structure
interpreting the symbols in $\Sigma$ and assignments to variables that
vary along time.  More specifically, a state $s=\langle M,\mu\rangle$
is given by a first-order structure $M$ and an assignment $\mu$ of
variables of $V$ into the domain of $M$. Given a state $s=\langle
M,\mu\rangle$ and a symbol $c$ of $\Sigma$ or variable $x\in V$ we use
$s(c)$ to denote the interpretation of $c$ in $M$ and $s(x)$ to denote
the value $\mu(x)$ assigned by $\mu$ to $x$.
Given $M$, let $V^M$ be the set of states with first-order structure
$M$. A trace $\sigma=\mktuple{M,\tau,\overline{\mu}}$ is given by a
first-order structure $M$, a time model $\tau$, and a mapping
$\overline{\mu}$ from the domain of $\tau$ into $V^M$. Given a trace
$\sigma=\mktuple{M,\tau,\overline{\mu}}$ and $t\in \tau$, we denote by
$\sigma(t)$ the state $\mktuple{M,\overline{\mu}(t)}$.

We assume to be given a $\Sigma$ first-order theory
$\mathcal{T}$. Given a $\Sigma$ first-order structure $M$, an
assignment $\mu$ to variables of $V$, and a $\Sigma$ first-order
formula $\phi$ over $V$, we use the standard notion of
$\mktuple{M,\mu}\models_\mathcal{T} \phi$. In the rest of the paper,
we omit the first-order signature $\Sigma$ and theory $\mathcal{T}$
for simplicity.

In our definition of trace, the first-order structure $M$ is shared by
all time points, meaning that the interpretation of the symbols in the
signature $\Sigma$ is rigid, does not vary with time. However, note
that the interpretation of symbols may not be ``fixed'' by the
background theory. These``uninterpreted'' symbols are also called
\emph{parameters}. For example, the signature $\Sigma$ can include the
symbols of the theory of reals (including the constants $0$ and $1$)
and an additional constant symbol $p$, whose value is not determined
by the theory but does not vary with time (thus, $p$ is a parameter).

Given a trace $\sigma=\mktuple{M,\tau,\overline{\mu}}$, a time point
$t$ of $\tau$, and a $\Sigma$ formula $\phi$, we define
$\sigma,t\models \phi$ recursively on the structure of $\phi$.
\begin{align*}
&\sigma,t\models p\text{ iff } \sigma(t)\models p\\
&\sigma,t\models \phi_1\wedge\phi_2\text{ iff }\sigma,t\models\phi_1
\text{ and }\sigma,t\models\phi_2\\
&\sigma,t\models\neg\phi\text{ iff }\sigma,t\not\models\phi\\
&\sigma,t\models\phi_1 \tilde{U}\phi_2\text{ iff there exists }
t'>t, \sigma,t'\models\phi_2\text{ and for all }t'',
t<t''<t',\sigma,t''\models\phi_1\\
&\sigma,t\models\phi_1 \tilde{S}\phi_2\text{ iff there exists }
t'<t, \sigma,t'\models\phi_2\text{ and for all
}t'', t'<t''<t,\sigma,t''\models\phi_1
\end{align*}
\noindent
Note that we are using the strict version of the ``until'' and
``since'' operators, where both arguments are required to hold in
points strictly greater or less than the current time.

Finally, $\sigma\models\phi$ iff $\sigma,\zero\models\phi$. We say
that $\phi$ is satisfiable iff there exists $\sigma$ such that
$\sigma\models\phi$. We say that $\phi$ is valid iff, for all $\sigma$,
$\sigma\models\phi$.

We use the following standard abbreviations:
\begin{align*}
&\phi_1 \vee\phi_2:=\neg(\neg\phi_1\wedge\neg\phi_2) &
&\top:=p\vee\neg p\\
&\bot:=\neg \top &
&\phi_1U\phi_2:=\phi_2\vee (\phi_1\wedge \phi_1\tilde{U}\phi_2)\\
&F\phi:=\top U\phi &
&G\phi:=\neg(F\neg\phi)\\
&\phi_1S\phi_2:=\phi_2\vee (\phi_1\wedge \phi_1\tilde{S}\phi_2)&
&P\phi:=\top S\phi\\
&H\phi:=\neg(P\neg\phi)
\end{align*}

As usual in many works on real-time temporal logics (e.g.,
~\cite{AFH96,Rab-LogComp98}), we assume the ``finite variability'' of
traces, i.e., that the evaluation of predicates by a trace changes
from true to false or vice versa only finitely often in any finite
interval of time. This can be lifted to temporal formulas in the sense
that temporal operators preserve the finite variability property (as
proved for example in \cite{AFH96}). Formally, we say that a trace
$\sigma$ is \emph{fine} for $\phi$ in a time interval $I$ iff for all
$t,t'\in I$, $\sigma,t\models \phi$ iff $\sigma,t'\models \phi$. A
trace $\sigma$ has the \emph{finite variability} property iff for
every formula $\phi$ there exists a sequence of points
${t}_0,{t}_1,t_2\ldots$ of $\sigma$ such that $\sigma$ is \emph{fine}
for $\phi$ in every interval $(t_i,t_{i+1})$, for $i\geq 0$.
In the following, we assume that traces have the
finite variability property.

\subsection{Next Operator and Function}
\label{sec:next}

Since $\tilde{U}$ is the strict version of the ``until'' operator, we
can write the standard $X$ as abbreviation:
\begin{align*}
X\phi&:=\bot \tilde{U}\phi
\end{align*}

$X$ is well defined in the different time models, also in the case of
dense or super-dense time. In the case of weakly-monotonic time,
$X\phi$ can be true only on a discrete step (i.e., in $\mktuple{n,t}$
if $\mktuple{n+1,t}$ is also in $T$). In the case of
strictly-monotonic time, $X\phi$ is always false.

With the strict until, we can also define a continuous counterpart of
the $X$ operator:
\begin{align*}
\tilde{X}\phi&:=\phi \tilde{U}\top\wedge \neg X\top
\end{align*}
\noindent
Note that $X\top$ is true in all and only in discrete steps. Thus,
$\tilde{X}\phi$ is always false in the case of discrete time, while in
the case of dense time it is true in the time points with a right
neighborhood satisfying $\phi$. In the super-dense time case,
$\tilde{X}\phi$ is false in the discrete steps, while in other time
points it is true if a right neighborhood satisfies $\phi$.  Note that
this is a variant of the more standard $\phi \tilde{U}\top$ formula,
which has been studied for example in \cite{FR07}. However, it was
considered only in the dense time case. Here, we added $\neg X\top$,
because it will be more convenient in the super-dense time case.

Similarly, we define the ``yesterday'' operators
$Y\phi:=\bot\tilde{S}\phi$ and $\tilde{Y}\phi:=\phi\tilde{S}\top\wedge
\neg Y\top$. We also define the weaker version of ``yesterday'' that
is true in the initial state:
$Z\phi:=(Y\top\vee\tilde{Y}\top)\rightarrow Y\phi$ and
$\tilde{Z}\phi:=(Y\top\vee\tilde{Y}\top)\rightarrow \tilde{Y}\phi$.

In the discrete-time setting, we often use also the functional
counterpart of $X$, here denoted by $next$~\cite{MP92}. Given a term
$u$, the interpretation of $next(u)$ in a trace $\sigma$ at the time
point $t$ is equal to the value of $u$ assigned by $\sigma$ at the
time point $t+1$. ``next'' does not typically have a counterpart in
the dense time case. Let LTL-next be the extension of LTL with
the $next$ function (with discrete time).

\subsection{Metric Temporal Operators}

In this section, we define some extensions of LTL that use metric
operators to constrain the time interval between two or more points.
We give a general version in the first-order setting that include also
weakly-monotonic time and parametric intervals.

Metric Temporal Logic (MTL) formulas are built with the following
grammar:
\begin{align*}
&\phi:=p(u,\ldots,u)\mid \phi\wedge\phi\mid \neg\phi\mid \phi
  \tilde{U}_I\phi\mid \phi \tilde{S}_I\phi\\
&I:=[cu,cu]\mid (cu,cu]\mid [cu,cu)\mid (cu,cu)\mid [cu,\infty)\mid (cu,\infty)\\
&cu:=c\mid f(cu,\ldots,cu)
\end{align*}
where the terms $u$ are defined as before and $cu$ are terms that do
not contain variables. Thus, the bounds of intervals used in MTL (as
well in the other logics defined below) are rigid and may contain
parameters. We assume here that the background first-order theory
contains the theory of reals and that the terms $cu$ have real type.

The abbreviations $\tilde{F}_I, \tilde{G}_I, \tilde{P}_I, \tilde{H}_I$
and their non-strict versions are defined in the usual way.
Moreover, for all logics defined in this section, we abbreviate the intervals
$[0,a]$, $[0,a)$, $[a,\infty)$, $(a,\infty)$, $[a,a]$, by respectively
    $\leq a$, $<a$, $\geq a$, $>a$, $=a$. Thus, for example,
    $\tilde{F}_{=p}b$ is an abbreviation of $\tilde{F}_{[p,p]}b$.

Let $\sigma=\mktuple{M,\tau,\overline{\mu}}$. We give the semantics
just for the metric operators:
\begin{align*}
&\sigma,t\models\phi_1 \tilde{U}_I\phi_2\text{ iff there exists } t'>t,
v(t')-v(t)\in M(I),\sigma,t'\models\phi_2\text{ and for all }t'',
t<t''<t',\sigma,t''\models\phi_1\\
&\sigma,t\models\phi_1 \tilde{S}_I\phi_2\text{ iff there exists } t'>t,
v(t)-v(t')\in M(I),\sigma,t'\models\phi_2\text{ and for all }t'',
t'<t''<t,\sigma,t''\models\phi_1
\end{align*}
\noindent where $M(I)$ is the set obtained from $I$ by substituting
the terms at the endpoints with their interpretation (thus it may be
also an empty set).

MTL$_0^\infty$ is the subset of MTL where the intervals in
metric operators are in the form $[0,a]$, $(0,a]$, $[0,a)$, $(0,a)$,
        $[a,\infty)$, $(a,\infty)$.

Event-Clock Temporal Logic (ECTL) is instead defined with the
following grammar:
\begin{align*}
&\phi:=p(u,\ldots,u)\mid \phi\wedge\phi\mid \neg\phi\mid \phi
  \tilde{U}\phi \mid \phi\tilde{S}\phi \mid \rhd_I\phi\mid
  \lhd_I\phi
\end{align*}
\noindent where $u$ and $I$ are defined as above.

We just give the semantics for the new symbols:
\begin{align*}
&\sigma,t\models\rhd_I\phi\text{ iff there exists } t'>t,
v(t')-v(t)\in M(I),\sigma,t'\models\phi\text{ and for all }t'',
t<t''<t',\sigma,t''\not\models\phi\\
&\sigma,t\models\lhd_I\phi\text{ iff there exists } t'>t,
v(t)-v(t')\in M(I),\sigma,t'\models\phi\text{ and for all }t'',
t'<t''<t,\sigma,t''\not\models\phi
\end{align*}

Finally, we define the Temporal Logic with Counting (TLC) with the
following grammar:
\begin{align*}
&\phi:=p(u,\ldots,u)\mid \phi\wedge\phi\mid \neg\phi\mid \phi
  \tilde{U}\phi \mid \phi\tilde{S}\phi \mid \overrightarrow{C}^k_{<cu}\phi\mid
  \overleftarrow{C}^k_{<cu}\phi
\end{align*}
\noindent where $u$ and $cu$ are defined as above.

We just give the semantics for the new symbols:
\begin{align*}
\sigma,t\models \overrightarrow{C}^k_{<cu}(\phi)&\text{ iff there exist }
t_1,\ldots,t_k, t<t_1<t_2<\ldots <t_k, v(t_k)-v(t)< M(cu) \\&\text{ such
  that for all } i\in [1,k], \sigma,t_i\models\phi\\
\sigma,t\models \overleftarrow{C}^k_{<cu}(\phi)&\text{ iff there exist }
t_1,\ldots,t_k, t_k<t_{k-1}<\ldots <t_1<t, v(t)-v(t_k)< M(cu) \\&\text{ such
  that for all } i\in [1,k], \sigma,t_i\models\phi
\end{align*}

\section{LTL with Event Freezing Functions}
\label{sec:main}

\subsection{Until next occurrence}

Before introducing the new operators, we observe some subtleties of
the dense-time semantics. In the discrete-time setting, $F\phi$ and
$(\neg\phi) U\phi$ are equivalent. In other words, if $\phi$ is true in
the future, there exists a first point in which it is true, while
$\phi$ is false in all preceding points. This is not the case in the
dense-time setting, since for example the third trace of
Figure~\ref{fig:atF} satisfies $F\phi$ but not $(\neg\phi) U\phi$: for
every time in which $\phi$ holds, there exists a left open interval
in which $\phi$ holds as well.

We can instead use another variant of the until operator defined as:
$$\phi_1U_C\phi_2:=\phi_1U(\phi_2\vee (\phi_1\wedge\tilde{X}\phi_2))$$

Thus, with $U_C$ we are requiring that $\phi_2$ holds in a point or in
every point of a right interval. In this case, we are guaranteed that
there exists a minimum point that satisfies such condition. In fact,
since we are assuming finite variability, $F\phi$ is equivalent to
$(\neg\phi) U_C\phi=(\neg \phi) U (\phi\vee \tilde{X}\phi)$. In the next
sections, we will use this condition to characterize the next point in
the future that satisfies $\phi$. In particular, when we say ``the
next point in the future in which $\phi$ holds'', we actually mean
$\phi$ holds in that point or in a right left-open interval (see also
Figure~\ref{fig:atF}).  Similarly, for the past case.

\begin{figure}[t]
\begin{center}
\includegraphics[scale=0.35]{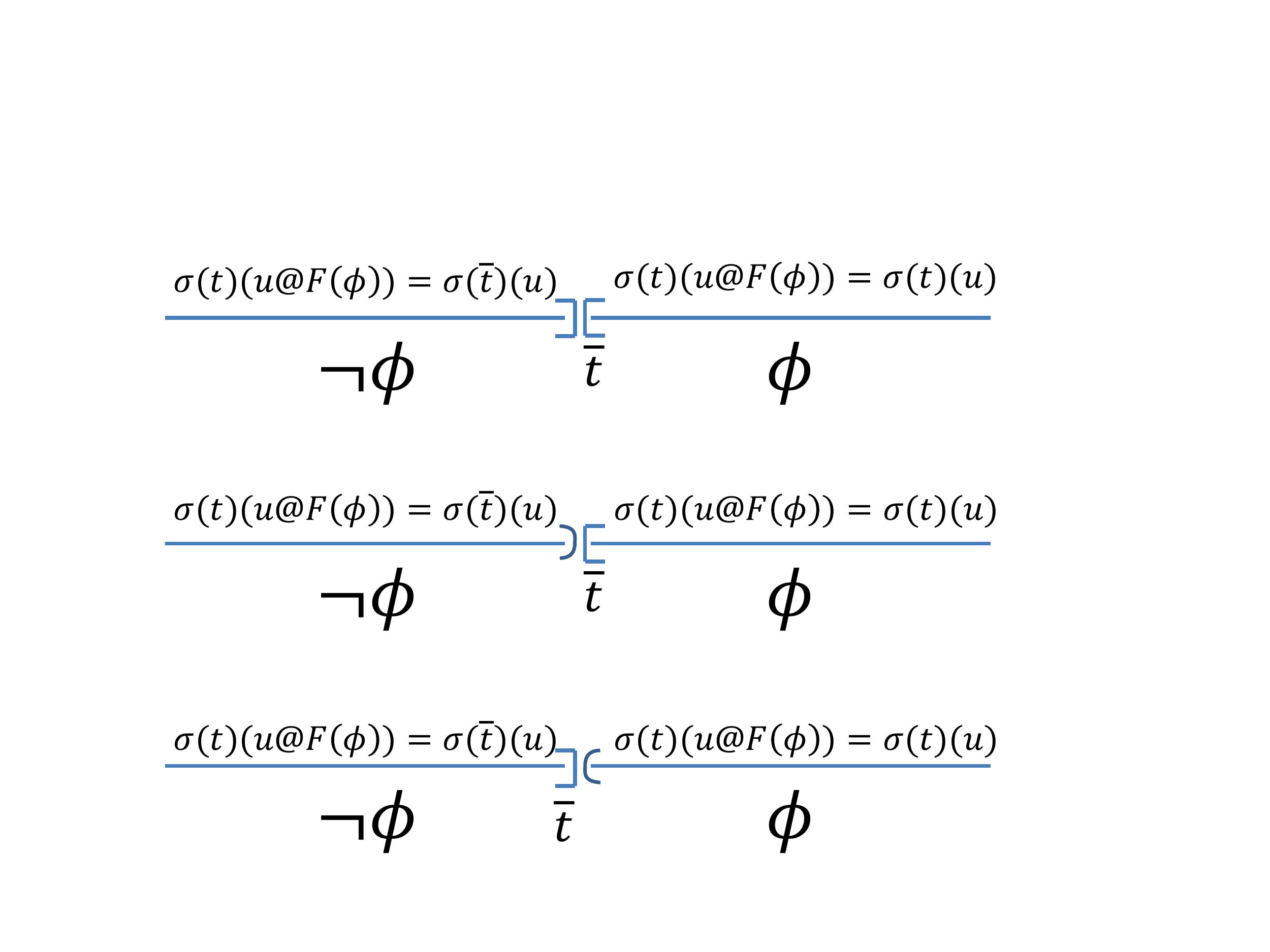}
\caption{Graphical view of different cases in which $\phi$ holds in
  the future. $\overline{t}$ represents ``the next point in the future
  in which $\phi$ holds''.}
\vspace{-3mm}
\label{fig:atF}
\end{center}
\end{figure}

Note that this is related to the issue of $U$ in the dense time
setting raised first by Bouajjani and Lakhnech in \cite{BouajjaniL95}
and later by Raskin and Schobbens in \cite{RaskinS97}, namely that
$\phi_1U\phi_2$ is satisfied only if the time interval in which
$\phi_2$ holds is left-closed. In \cite{BouajjaniL95}, this is solved
by considering $(\phi_1\vee\phi_2)U\phi_2$. However, this does not
solve our issue of characterizing the first point in which $\phi_2$
holds. In \cite{RaskinS97}, the issue was solved at the semantic level
by defining the $U$ on timed state sequence that are fine for the
subformulas and quantifying over the time intervals of the sequence
instead of over the points of the time domain. We instead chose a more
classical approach to define the semantics which seems to clarify
better what we mean for ``the next point in which $\phi$ holds''. This
is more similar to the semantics defined in \cite{HenzingerRS98} for
event clocks in Event-Clock Timed Automata and for the corresponding
quantifiers in the equally expressive monadic logic. However, in
\cite{HenzingerRS98}, a nonstandard real number is used in case $\phi$
holds in a left-open interval.

\subsection{Event Freezing Functions}

We extend the logic with two binary operators, ``\atfuture''
$u\atsF(\phi)$ and ``\atpast'' $u\atsP(\phi)$,
which take in input a term $u$ and a formula $\psi$ and represent the
value of $u$ at the next point in the future, respectively at the last point
in the past, in which $\psi$ holds. If such point does not exist we
consider a default value represented by a constant
$\default_{u\atsF(\psi)}$ or $\default_{u\atsP(\psi)}$. As in SMT, we also use
an if-then-else operator, extended to the temporal case.

The set of LTL with Event Freezing Functions (LTL-EF) formulas is therefore
defined as follows:
\begin{align*}
&\phi:=p(u,\ldots,u)\mid \phi\wedge\phi\mid \neg\phi\mid \phi
  \tilde{U}\phi\mid \phi \tilde{S}\phi\\
&u:=c\mid x\mid f(u,\ldots,u)\mid u\atsF(\phi) \mid u\atsP(\phi) \mid
  \ite(\phi,u,u)
\end{align*}

A formula $\phi$ is interpreted on a trace
$\sigma=\mktuple{M,\tau,\overline{\mu}}$, where $M$ is a first-order
  structure over the signature extended with the constant symbols
  $\default_u$ for every event freezing term $u$ in $\phi$. The semantics
  of LTL is thus extended as follows:
\begin{itemize}
\item
$\sigma(t)(u\atsF(\phi))=\sigma(t')(u)$ if there exists $t'> t$ such
  that, for all $t''$, $t< t''<t'$, $\sigma,t''\not\models\phi$ and
  $\sigma,t'\models\phi$; $\sigma(t)(u\atsF(\phi))=\sigma(t')(u)$ if there exists $t'\geq t$ such
  that, for all $t''$, $t< t''\leq t'$, $\sigma,t''\not\models\phi$ and
  $\sigma,t'\models\tilde{X}\phi$; otherwise,
  $\sigma(t)(u\atsF(\phi))=M(\default_{u\atsF(\phi)})$
\item
$\sigma(t)(u\atsP(\phi))=\sigma(t')(u)$ if there exists $t'< t$ such
  that, for all $t''$, $t'< t''< t$, $\sigma,t''\not\models\phi$ and
  $\sigma,t'\models\phi$;
  $\sigma(t)(u\atsP(\phi))=\sigma(t')(u)$ if there exists $t'\leq t$ such
  that, for all $t''$, $t'\leq t''< t$, $\sigma,t''\not\models\phi$ and
  $\sigma,t'\models\tilde{Y}\phi$; otherwise,
  $\sigma(t)(u\atsP(\phi))=M(\default_{u\atsP(\phi)})$
\item
$\sigma(t)(\ite(\phi,u_1,u_2))=\sigma(t)(u_1)$ if $\sigma,t\models\phi$,
  else $\sigma(t)(\ite(\phi,u_1,u_2))=\sigma(t)(u_2)$
\end{itemize}

The ``if-then-else'' operator $ite$ can be used to define
the non-strict version:
\begin{align*}
u\atF(\phi):=&\ite(\phi,u,u\atsF(\phi))\\
u\atP(\phi):=&\ite(\phi,u,u\atsP(\phi))
\end{align*}

We define the following abbreviations:
\begin{align*}
u\atsF^1(\phi):=&u\atsF(\phi)&
u\atsF^{i+1}(\phi):=&(u\atsF(\phi))\atsF^{i}(\phi)\text{ for }i\geq 1\\
u\atsP^1(\phi):=&u\atsP(\phi)&
u\atsP^{i+1}(\phi):=&(u\atsP(\phi))\atsP^{i}(\phi)\text{ for }i\geq 1
\end{align*}

\subsection{Extension with Explicit Time}

In this section, we extend the language with an explicit notion of
time that can be constrained using the event freezing functions
defined above. In particular, we introduce an explicit symbol $time$,
which represents the time elapsed from the initial state. We allow
$time$ to be compared with constant terms.

The new set of LTL-EF formulas with explicit time (XLTL-EF) is defined
as follows:
\begin{align*}
&\phi:=p(u,\ldots,u)\mid tu \bowtie cu\mid \phi\wedge\phi\mid \neg\phi\mid \phi
  \tilde{U}\phi\mid \phi \tilde{S}\phi\\
&u:=c\mid x\mid f(u,\ldots,u)\mid u\atsF(\phi) \mid u\atsP(\phi) \mid
  \ite(\phi,u,u)\\
&tu:=time \mid tu\atsF(\phi) \mid tu\atsP(\phi)\\
&cu:=c\mid f(cu,\ldots,cu)\\
&\bowtie:=<\mid >\mid \leq \mid \geq
\end{align*}

The semantics of LTL-EF is extended as follows:
$\sigma(t)(time):=v(t)$

Note that we assume that the signature $\Sigma$ contains the real
arithmetic operators and that the underlying theory contains the
theory of reals.

\subsection{Coverage of Metric Operators}

These operators can be seen as a generalization of the ECTL
operators as below:
\begin{align*}
&\rhd_I\phi:=time\atsF(\phi)-time\in I\wedge \neg\phi\tilde{U}\phi\\
&\lhd_I\phi:=time-time\atsP(\phi)\in I\wedge \neg\phi\tilde{S}\phi
\end{align*}
We can encode similarly MTL$_0^\infty$ operators. As proved in
\cite{HenzingerRS98}, in case of non-singular intervals with real
constant bounds, MTL operators can be expressed in ECTL (and thus in
XLTL-EF).

We can also express TLC properties as follows:
\begin{align*}
&\overrightarrow{C}^k_{<cu}(\phi):=time\atsF^k(\phi)-time<cu\wedge \tilde{F}^k(\phi)\\
&\overleftarrow{C}^k_{<cu}(\phi):=time-time\atsP^k(\phi)<cu\wedge \tilde{P}^k(\phi)
\end{align*}

\subsection{Sensor Example}

Consider a sensor with input $y$ and output $x$ and a Boolean flag
$correct$ that represents whether or not the value reported by the sensor is
correct. Let us specify that the output $x$ is always equal to the
last correct input value with $G(x=y\atP(correct))$. We assume that a
failure is permanent: $G(\neg correct\rightarrow G\neg
correct)$. Consider also a Boolean variable $read$ that represents the
event of reading the variable $x$. Let us say that the reading happens
periodically with period $p$: $p>0 \wedge read \wedge G(read
\rightarrow \rhd_{=p}read)$. Finally, let us say that an alarm $a$ is
true if and only if the last two read values are the same:
$G(a\leftrightarrow x\atsP(read)=x\atsP^2(read))$.

We would like to prove that, given the above scenario, every point in which
the sensor is not correct is followed within $2* p$ by an
alarm:
\begin{align*}
&(G (x=y\atP(correct)) \wedge
G(\neg correct\rightarrow G\neg correct)\wedge\\
&p>0 \wedge read \wedge G(read \rightarrow\rhd_{=p}read)\wedge
G(a\leftrightarrow x\atsP(read)=x\atsP^2(read)))\\
&\rightarrow G(\neg correct\rightarrow F_{\leq 2*p}a)
\end{align*}

In the following, we show that this kind of problems can be indeed
solved automatically with SMT-based techniques.

\section{Satisfiability Procedure}
\label{sec:sat}

\subsection{Overview of the Procedure}

The satisfiability problem for (first-order) LTL with discrete time,
and thus also for (X)LTL-EF which is an extension thereof, is in
general undecidable (see for example~\cite{GhilardiNRZ07}). However,
we can reduce it to SMT-based model checking, which although
undecidable has effective and mature tool support (also in the case of
LTL-next). We thus propose to reduce the satisfiability for (X)LTL-EF
to the one for LTL-next. The approach consists of the following steps:
1) the formula is translated into an equisatisfiable one with
discrete-time model, 2) the event freezing functions are removed
generating an equisatisfiable LTL-next formula.
Since LTL-EF is a subset of XLTL-EF, in the following we consider just
XLTL-EF formulas.

\subsection{Discretization}

Given an XLTL-EF formula with dense or super-dense time, we
create an equisatisfiable one with discrete time. We will use the
non-strict version of the temporal operators as these are those
typically supported by tools for LTL.
The discretization approach is similar to the one described
in~\cite{CAV09}. The idea is to split the time evolution is a sequence
of singular or open intervals in such a way that the trace is fine for the
input formula on such intervals. To this purpose
we introduce three variables that encode
a sequence of time intervals used to sample the value of variables:
$\iota$ is a Boolean variable that encodes if the interval is singular
or open; $\delta$ is a real variable that encodes the time elapsed
between two samplings; $\zeta$ is a real variable that accumulates
arbitrary sums of $\delta$. A constraint $\psi_\iota$ ensures that the
value of these additional variables represent a valid time interval
sequence (e.g., after an open interval there must be a singular
interval and $\zeta$ is infinitely often greater than $1$ and reset).
Another constraint $\psi_{time}$ ensures that the variable $time$ is
equal to the accumulation of $\delta$ and that the evaluation of
predicates in the formula is uniform in open intervals.

Given these extra variables, we can define the translation. Given a
formula $\phi$ over $V$, we rewrite $\phi$ into $\phi_D$ over
$V\cup\{\iota,\delta, \zeta\}$ defined as:
\begin{align*}
\phi_D:= &\mathcal{D}(\phi)\wedge \psi_\iota\wedge \psi_{time}
\end{align*}
\noindent
where $\mathcal{D}$, $\psi_\iota$, and $\psi_{time}$ are defined as
follows.

$\mathcal{D}(\phi)$ is defined recursively on the structure of $\phi$ and
rewrites the temporal operators splitting between the case in which
the current interval is singular ($\iota$) or open ($\neg \iota$).

Let us consider first $\phi_1\tilde{U}\phi_2$; intuitively, to hold in
a time point $t$, if $t$ belongs to an open interval (fine for
$\phi_1$), then $\phi_1$ must hold in $t$; similarly, if
$\phi_1\tilde{U}\phi_2$ because $\phi_2$ holds in $t'>t$ and $t'$ is
part of an open interval, also $\phi_1$ must hold in $t'$. Thus,
$\phi_1\tilde{U}\phi_2$ is translated as follows: either the current
interval is open ($\neg\iota$), $\phi_2$ holds in a future singular
interval and $\phi_1$ holds now and until that interval; or the
current interval is open, $\phi_2$ holds in a future open interval and
$\phi_1$ holds now and until that interval included; or the current
interval is singular, $\phi_2$ holds in a future singular interval and
$\phi_1$ holds (strictly) until that interval; or the current interval
is singular, $\phi_2$ holds in a future open interval and $\phi_1$
holds (strictly) until that interval included. Similarly for the past
case. Overall:
\begin{align*}
\mathcal{D}(\phi_1\tilde{U}\phi_2):=&
(\neg\iota\wedge \mathcal{D}(\phi_1)\wedge (\mathcal{D}(\phi_1)U
((\iota\wedge \mathcal{D}(\phi_2)) \vee (\mathcal{D}(\phi_1)\wedge \mathcal{D}(\phi_2)))))\vee\\
&(\iota\wedge X(\mathcal{D}(\phi_1)U ((\iota\wedge \mathcal{D}(\phi_2)) \vee (\mathcal{D}(\phi_1)\wedge \mathcal{D}(\phi_2)))))\\
\mathcal{D}(\phi_1\tilde{S}\phi_2):=&
(\neg\iota\wedge \mathcal{D}(\phi_1)\wedge \mathcal{D}(\phi_1)S
((\iota\wedge \mathcal{D}(\phi_2)) \vee (\mathcal{D}(\phi_1)\wedge \mathcal{D}(\phi_2))))\vee\\
&(\iota\wedge Y(\mathcal{D}(\phi_1)S ((\iota\wedge \mathcal{D}(\phi_2)) \vee (\mathcal{D}(\phi_1)\wedge \mathcal{D}(\phi_2)))))
\end{align*}

Let us consider now $u\atsF(\phi)$. We first define
$\mathcal{D}(u\atF(\phi))$, which is used as intermediate step to
define $\mathcal{D}(u\atsF(\phi))$. The discretization of $u\atF(\phi)$ is
translated into the value of $u$ at the first state in the future such
that $\phi$ holds in that state or the following state corresponds to
an open interval in which $\phi$ holds:
\begin{align*}
\mathcal{D}(u\atF(\phi)):=&\mathcal{D}(u)\atF(\mathcal{D}(\phi)\vee X(\neg\iota\wedge \mathcal{D}(\phi)))\\
\mathcal{D}(u\atP(\phi)):=&\mathcal{D}(u)\atP(\mathcal{D}(\phi)\vee Y(\neg\iota\wedge \mathcal{D}(\phi)))
\end{align*}

If the current interval is open and $\phi$ holds in all points of the
interval, then $u\atsF(\phi)=u=u\atF(\phi)$. Similarly, if the current
interval is singular and is followed by an open interval in which
$\phi$ holds, then $u\atsF(\phi)=u=u\atF(\phi)$. If the current interval
is open and $\phi$ does not hold in the interval, then again
$u\atsF(\phi)=u\atF(\phi)$. Finally,
if the interval is singular and is not followed by an open interval in
which $\phi$ holds, $u\atsF(\phi)$ is equal to the value of
$u\atF(\phi)$ in the next interval. The overall translation is the
following:
\begin{align*}
\mathcal{D}(u\atsF(\phi)):=&\ite(\iota\wedge X(\iota\vee\neg\mathcal{D}(\phi)),next(\mathcal{D}(u\atF(\phi))),\mathcal{D}(u\atF(\phi)))\\
\mathcal{D}(u\atsP(\phi)):=&\ite(\iota\wedge Z(\iota\vee\neg\mathcal{D}(\phi)),prev(\mathcal{D}(u\atP(\phi))),\mathcal{D}(u\atP(\phi)))
\end{align*}
\noindent where we use $prev$ for simplicity with the following
semantics: $\sigma(0)(prev(u\atP(\phi)))=\default_{u\atP(\phi)}$ and
$\sigma(i+1)(prev(u\atP(\phi)))=\sigma(i)(u\atP(\phi))$. In practice,
this is rewritten in terms of $next$ and an extra monitor variable in
the usual way.

The following completes the definition with the trivial cases:
\begin{align*}
\mathcal{D}(\phi_1\wedge\phi_2):=&\mathcal{D}(\phi_1)\wedge\mathcal{D}(\phi_2)&
\mathcal{D}(\neg\phi_1):=&\neg\mathcal{D}(\phi_1)\\
\mathcal{D}(p(u_1,\ldots,u_n)):=&p(\mathcal{D}(u_1),\ldots,\mathcal{D}(u_n))&
\mathcal{D}(tu\bowtie cu):=&\mathcal{D}(tu)\bowtie\mathcal{D}(cu)\\
\mathcal{D}(f(u_1,\ldots,u_n)):=&f(\mathcal{D}(u_1),\ldots,\mathcal{D}(u_n))&
\mathcal{D}(time):=&time\\
\mathcal{D}(c):=&c&
\mathcal{D}(v):=&v
\end{align*}

$\psi_\iota$ encodes the structure of the time model (to enforce for
example that after an open interval there must be a singular one and
that in a discrete step time does not elapse):
\begin{align*}
\psi_\iota:=&\iota \wedge G ( (\iota \wedge \delta=0 \wedge X(\iota) ) \vee
     (\iota \wedge \delta>0 \wedge X(\neg\iota) ) \vee
     (\neg \iota \wedge \delta>0 \wedge X(\iota) ) ) \wedge\\
& G ((next(\zeta)-\zeta=\delta) \vee (\zeta\geq 1\wedge next(\zeta)=0)) \wedge
  G F (\zeta\geq 1 \wedge next(\zeta)=0)
\end{align*}

Finally, $\psi_{time}$ encodes the value of $time$ and forces the
uniformity of predicates over $time$ in open intervals:
\begin{align*}
\psi_{time}:=&time=0\wedge G(next(time)-time=\delta)\wedge\\
&\bigwedge_{tu\bowtie cu\in Sub(\phi)}
G (\neg\iota\rightarrow
((\mathcal{D}(tu\leq cu) \rightarrow X \mathcal{D}(tu\leq cu))\wedge
(\mathcal{D}(tu\geq cu) \rightarrow Y \mathcal{D}(tu\geq cu))))
\end{align*}
\noindent
where $Sub(\phi)$ denotes the set of subformulas of $\phi$.

Note in particular that we require to split the time intervals in such
a way that
for every constant $cu$ occurring in a time constraint, $[cu,cu]$ is a
time interval in the sequence. Note that $cu$ can be in general a term
built with the signature symbols that are interpreted rigidly.

Written as above the discretization clearly produces a formula whose
size is exponential in the input. However, since we are interested in
equisatisfiability we can always use extra variables (one for
subformula) to obtain a linear-size formula.

We now prove that the translation is correct, i.e., that the new
formula is equisatisfiable.

\begin{theorem}
$\phi$ and $\phi_D$ are equisatisfiable.
\end{theorem}

\begin{proof}
Given a trace $\sigma=\mktuple{M,\tau,\overline{\mu}}$ satisfying
$\phi$ we can build a trace $\sigma_D$ with a discrete time model
satisfying $\phi_D$ as follows. Let $I_0,I_1,I_2,\ldots$ be a sequence
of time intervals such that 1) $\sigma$ is fine for all subformulas of
$\phi$ in each interval $I_i$, 2) each interval $I_i$ in the sequence
is singular or open, and 3) in case of super-dense time, for all
$i\geq 0$, there exists an integer $n_i$ such that
$\mktuple{n_i,t}\in\tau$ for all $t\in I_i$. We build an assignment to
$\iota$, $\delta$ and $\zeta$ based on such sequence. The values of
$\iota$, $\delta$, and $\zeta$ are determined by the sequence of
intervals in order to satisfy $\psi_\iota$.

Let us define the value assigned by $\sigma_D$ to $\iota,\delta,\zeta$ as
follows:
\begin{itemize}
\item
$\sigma_D(i)(\iota)=\top$ iff $I_i$ is singular;
\item
$\sigma_D(i)(\delta)=(r(I_{i+1})-l(I_{i+1}))/2$ if $I_i$ is singular\\
otherwise $\sigma_D(i)(\delta)=(r(I_{i})-l(I_i))/2$;
\item
$\sigma_D(0)(\zeta)=0$\\
$\sigma_D(i+1)(\zeta)=\sigma_D(i)(\zeta)+\sigma_D(i+1)(\delta)$ if
  $\sigma_D(i+1)(\zeta)\leq 1$\\
otherwise $\sigma_D(i+1)(\zeta)=0$
\end{itemize}

Thus, $\sigma_D\models\psi_\iota$. Notice in particular, that if
$I_i=I_{i+1}$ then $\delta=0$, if $I_i$ is singular and $I_{i+1}$ is
not then $\delta$ is equal to half of the length of $I_{i+1}$ and if
$I_i$ is not singular then $\delta$ is equal to half of the length of
$I_{i}$.

Let $\delta_i:=\sigma(i)(\delta)$ and $t_i=\sum_{0\leq h<i}\delta_h$
for all $i\geq 0$. Notice that for all $i\geq 0$, $t_i\in I_i$. We
complete the time model of $\sigma_D$ by defining $v(i):=t_i$ for all
$i\geq 0$.

Let $\textbf{t}_i=t_i$ in case $\sigma$ has a dense time and
$\textbf{t}_i=\mktuple{n_i,t_i}$ in case of super dense time. Let us
complete the definition of $\sigma_D$ by saying that for all $i\geq
0$, $\sigma_D(i)(x):=\sigma(\textbf{t}_i)(x)$.

We now prove that, for all $i\geq 0$, for all subformulas $\psi$ of
$\phi$, $\sigma,\textbf{t}_i\models
\psi$ iff $\sigma_D,i\models \mathcal{D}(\psi)$ and for all terms $u$
in $\phi$, $\sigma(\textbf{t}_i)(u)=\sigma_D(i)(\mathcal{D}(u))$.

By definition of $\atF$, $\sigma(\textbf{t}_i)(u\atF(\psi))$ is
the value of $u$ at the next point $t\geq\textbf{t}_i$ such that
$\sigma,t\models \psi\vee \tilde{X}\psi$. Since $\sigma$ is fine
for $\psi$, $t$ must belong to a singular interval $I_j$ with $j\geq i$
(so $t=\textbf{t}_j$). By inductive hypothesis,
$\sigma,\textbf{t}_j\models\psi$ iff $\sigma_D,j\models
\mathcal{D}(\psi)$ and $\sigma,\textbf{t}_{j+1}\models\psi$ iff
$\sigma_D,j+1\models \mathcal{D}(\psi)$. Thus,
$\sigma,\textbf{t}_j\models\psi\vee \tilde{X}\psi$ iff
$\sigma_D,j\models \mathcal{D}(\psi)\vee X (\neg\iota\wedge
\mathcal{D}(\psi))$.  Moreover, still by inductive hypothesis, for all
$i<h<j$, $\sigma,\textbf{t}_h\models \psi$ iff $\sigma_D,h\models
\mathcal{D}(\psi)$ and
$\sigma(\textbf{t}_j)(u)=\sigma_D(j)(\mathcal{D}(u))$. Thus,
$\sigma(\textbf{t}_i)(u\atF(\psi))=\sigma_D(i)(\mathcal{D}(u\atF(\psi)))$.

It is routine to prove the other cases and we can conclude that
$\sigma_D\models \mathcal{D}(\phi)$.

Finally, $\sigma_D\models\psi_{time}$: in fact, $\sigma_D\models
time=0\wedge G(next(time)-time=\delta)$ by definition of $\sigma_D$;
the rest of $\psi_{time}$ is trivially satisfied because $\sigma$ is
fine for $\phi$.

Vice versa, suppose that there exists $\sigma$ with discrete time such
that $\sigma\models\mathcal{D}(\phi)$. Then we can build a $\sigma_C$
with super-dense time such that $\sigma_C\models\phi$ as follows. Let
$t_i=\sum_{0\leq h<i}\delta_h$. $I_i:=[t_i,t_i]$ if
$\sigma,i\models\iota$; otherwise $I_i:=(t_{i-1},t_{i+1})$. Let
$\sigma(t)(v)=\sigma(i)(v)$ for every $t\in I_i$.

It is routine to prove that $\sigma_C\models\phi$.
\end{proof}

\subsection{Removing Event Freezing Functions}

In the following, we assume that satisfiability is restricted to
traces with discrete time and we use the non-strict version of temporal
operators. If the term $u\atF(\phi)$ occurs in a formula $\psi$, we
can obtain a formula $\mathcal{R}(\psi,u\atF(\phi))$ equisatisfiable
to $\psi$ where the term $u\atF(\phi)$ has been replaced with a fresh
variable $p_{u\atF(\phi)}$. More specifically,
\begin{align*}
\mathcal{R}(\psi,u\atF(\phi)):=&\psi[p_{u\atF(\phi)}/u\atF(\phi)]\wedge\\
  &G({F}\phi\rightarrow (\neg \phi \wedge
next(p_{u\atF(\phi)})=p_{u\atF(\phi)}) {U} (\phi \wedge
p_{u\atF(\phi)}=u)) \wedge\\
  &G({G}\neg\phi\rightarrow p_{u\atF(\phi)}=\default_{u\atF(\phi)})
\end{align*}

$\mathcal{R}(\psi,u\atF(\phi))$ is a formula on an extended set of
variables. Namely, if $\phi$ is a formula over variables $V$, then
$\mathcal{R}(\psi,u\atF(\phi))$ is a formula over
$V\cup\{p_{u\atF(\phi)}\}$, where $p_{u\atF(\phi)}$ does not occur in
$\psi$. However, the value of $p_{u\atF(\phi)}$ is uniquely determined
by a trace over $V$. In other words, given a trace $\sigma$ over $V$,
we can define a trace $\mathcal{R}(\sigma,u\atF(\phi))$ over
$V\cup\{p_{u\atF(\phi)}\}$ such that $\sigma\models\phi$ iff
$\mathcal{R}(\sigma,u\atF(\phi))\models
\mathcal{R}(\psi,u\atF(\phi))$.  $\mathcal{R}(\sigma,u\atF(\phi))$ is
simply defined as follows:
\begin{align*}
& \mathcal{R}(\sigma,u\atF(\phi))(t)(x)=\sigma(t)(x), x\in V\\
& \mathcal{R}(\sigma,u\atF(\phi))(t)(p_{u\atF(\phi)})=\sigma(t)(u\atF(\phi))
\end{align*}

\begin{theorem}
If $\sigma\models\phi$ then
$\mathcal{R}(\sigma,u\atF(\phi))\models
\mathcal{R}(\psi,u\atF(\phi))$. If $\sigma\models
\mathcal{R}(\psi,u\atF(\phi))$, then $\sigma\models\phi$.
Thus, $\psi$ and $\mathcal{R}(\psi,u\atF(\phi))$ are equisatisfiable.
\end{theorem}

\begin{proof}
Let us assume that $\sigma\models\phi$.

Given the definition of $\mathcal{R}(\sigma,u\atF(\phi))$, the
prophecy variable $p_{u\atF(\phi)}$ is given the value of the term
$u\atF(\phi)$, and thus $\mathcal{R}(\sigma,u\atF(\phi))\models
\psi[p_{u\atF(\phi)}/u\atF(\phi)]$.

For every $t$, if $\sigma,t\models {F}(\phi)$, then there
exists $t'\geq t$ such that $\sigma,t'\models \phi$ and for all $t''$,
$t\leq t''<t'$, $\sigma,t''\not\models\phi$. Thus,
$\sigma(t')(u\atF(\phi))=\sigma(t'')(u\atF(\phi))=\sigma(t')(u)$. Thus
$\sigma,t\models G({F}(\phi)\rightarrow(\neg \phi \wedge
next(p_{u\atF(\phi)})=p_{u\atF(\phi)}) {U} (\phi \wedge
p_{u\atF(\phi)}=u))$.

For every $t$, if $\sigma,t\models {G}(\neg\phi)$, then
$\sigma(t)(u\atF(\phi))=\sigma(t)(\default_{u\atF(\phi)})$. Thus,
$\sigma,t\models G({G}(\neg\phi)\rightarrow
p_{u\atF(\phi)}=\default_{u\atF(\phi)})$.

Vice versa, let us assume that $\sigma\models
\mathcal{R}(\psi,u\atF(\phi))$.
It is sufficient to prove that
$\sigma(t)(p_{u\atF(\phi)})=\sigma(t)(u\atF(\phi))$.

Let us assume that there exists $t'\geq t$ such that, for all $t''$,
$t\leq t''<t'$, $\sigma,t''\not\models\phi$ and
$\sigma,t'\models\phi$; thus,
$\sigma(t)(u\atF(\phi))=\sigma(t')(u)$. Since $\sigma,t\models
{F}\phi\rightarrow (\neg \phi \wedge
next(p_{u\atF(\phi)})=p_{u\atF(\phi)}) {U} (\phi \wedge
p_{u\atF(\phi)}=u)$, then $\sigma(t)(p_{u\atF(\phi)})=\sigma(t')(u)$
as well.

If such $t'$ does not exists, then
$\sigma(t)(u\atF(\phi))=\sigma(t)(\default_{u\atF(\phi)})$. Also,
$\sigma,t\not\models F(\phi)$. Since $\sigma,t\models
{G}\neg\phi\rightarrow p_{u\atF(\phi)}=\default_{u\atF(\phi)}$, then
$\sigma(t)(p_{u\atF(\phi)})=\sigma(t)(\default_{u\atF(\phi)})$  as well.
This concludes the proof.
\end{proof}

We similarly remove past event freezing operator $\atP$ with the
following rule:
\begin{align*}
\mathcal{R}(\psi,u\atP(\phi)):=&\psi[p_{u\atP(\phi)}/u\atP(\phi)]\wedge\\
  &G({P}\phi\rightarrow (\neg \phi \wedge
Z(next(p_{u\atP(\phi)})=p_{u\atP(\phi)})) {S} (\phi \wedge
p_{u\atP(\phi)}=u)) \wedge\\
  &G({H}\neg\phi\rightarrow p_{u\atP(\phi)}=\default_{u\atP(\phi)})
\end{align*}

\section{Experimental Evaluation}
\label{sec:eval}

The satisfiability procedure presented in the previous section has
been implemented in a simple prototype written in C that takes in
input a XLTL-EF formula and performs the described
transformations. Then, we use nuXmv~\cite{CavadaCDGMMMRT14} to
check LTL satisfiability modulo theories, in
particular we use the algorithm that combines IC3IA~\cite{ic3ia} with
k-liveness~\cite{kliveness}. Actually, we create an universal model and we
apply model checking to check that the formula is valid. We checked
the validity of different formulas using a machine with Intel(R)
Core(TM) i7-3720QM CPU at 2.60GHz with 4GB of memory. All results are
available at \url{https://es.fbk.eu/people/tonetta/papers/gandalf17/}.

In Table~\ref{tab-eval}, we report the time needed to solve some
example formulas. The sensor example described above was proved by
nuXmv in 16s. This result is promising considering that we do not
implement any optimization neither in the translation nor in the
engine. The reported time includes the time for translation,
verification, and counterexample generation in case of not valid
formulas.

As proof of concept, we also verify the validity of some
MTL$_0^\infty$ (e.g. $( G (a \rightarrow F_{\leq p} b) \wedge G (b
\rightarrow F_{\leq p}c) ) \rightarrow G (a \rightarrow F_{\leq 2* p}
c)$) and ECTL (e.g. $(\rhd_{=q}\rhd_{=p} b) \rightarrow (\rhd_{=p+q}b
\vee \rhd_{\leq q} b)$) formulas, also including parameters. To the
best of our knowledge, this is the first tool that is able to
automatically prove the validity of this kind of formulas.

\begin{table}[t]
\begin{center}
\begin{tabular}{|l|c|c|}
\hline
Formula & Valid & Time in sec.\\
\hline
$(G (x=y\atP(correct)) \wedge
G(\neg correct\rightarrow G\neg correct)\wedge$ & &\\
$p>0 \wedge read \wedge G(read \rightarrow\rhd_{=p}read)\wedge
G(a\leftrightarrow x\atsP(read)=x\atsP^2(read)))$ & &\\
$\rightarrow G(\neg correct\rightarrow F_{\leq 2*p}a)$
 & Yes & 16\\
\hline
$G ( b \rightarrow (x \atF(b)=x) )$ & Yes & 0\\
\hline
$G ( \tilde{X}(b) \rightarrow (x \atF(b)=x) )$ & Yes & 0\\
\hline
$F b \rightarrow (\neg b U (b \vee (\tilde{X} b)))$ & Yes & 0\\
\hline
$( G (a \rightarrow F_{\leq 1} b) \wedge G (b \rightarrow F_{\leq 1}c)
) \rightarrow G (a \rightarrow F_{\leq 2} c)$ & Yes & 5\\
\hline
$( G (a \rightarrow F_{\leq p} b) \wedge G (b \rightarrow F_{\leq p}c)
) \rightarrow G (a \rightarrow F_{\leq 2* p} c)$ & Yes &
8\\
\hline
$(\rhd_{=q}\rhd_{=p} b) \rightarrow (\rhd_{=p+q}b \vee \rhd_{\leq q} b)$ & Yes & 5 \\
\hline
$F b \rightarrow (\neg b U b)$ & No & 0 \\
\hline
$\neg (x=y \wedge F_{\geq 3} x>y)$ & No & 0\\
\hline
$( G (a \rightarrow F_{\leq 3} b) \wedge G (b \rightarrow F_{\leq 3}c)
) \rightarrow G (a \rightarrow F_{\leq 3} c)$ & No & 2\\
\hline
$(\rhd_{=p}\rhd_{=p} b) \rightarrow (\rhd_{=2*p}b)$ & No & 4\\
\hline
\end{tabular}
\end{center}
\caption{Some examples and their verification results}
\label{tab-eval}
\end{table}

\section{Conclusions and Future Work}
\label{sec:conc}

In this paper, we considered an extension of first-order linear-time
temporal logic with two new event freezing functional symbols, which
represent the value of a term at the next state in the
future or last state in the past in which a formula holds. We defined
the semantics in different time models considering discrete time with
real timestamps, dense and super-dense (weakly-monotonic) time. We
precisely characterized what we mean for ``next point in the future in
which $\phi$ holds'' so that, assuming finite variability, such point
always exists also in the dense time setting. Using an explicit
variable that represents time, standard metric operators can be encoded
in the new logic. We provided a reduction to equisatisfiable
discrete-time formulas without event freezing functions and we solve
the satisfiability of the latter by SMT-based model checking. A
prototype implementation of the technique shows that the approach can
analyze interesting properties in an automated way despite the
expressiveness of the logic.

The directions for future works are manifold. We want to integrate
this techniques in mature tools such as nuXmv~\cite{CavadaCDGMMMRT14} and in
OCRA~\cite{CDT13} for contract-based reasoning; we want to extend it
to encompass variables with continuous function evolution and
constraints on their derivatives as in HRELTL~\cite{CAV09}
(HyCOMP~\cite{hycomp} is actually already supporting
$time\_until(\phi)$ and $time\_since(\phi)$, which are a
restricted version of $time\atF(\phi)$ and $time\atP(\phi)$, where
$\phi$ must represent a discrete change); finally, we want to apply
the new logic in
industrial use cases within the CITADEL project
(\url{http://citadel-project.org/}) to specify complex properties of
monitoring components.

\bibliographystyle{eptcs}
\bibliography{main}

\end{document}